\documentclass[11pt,oneside]{article}

\usepackage{palatino}
\usepackage[margin=1in]{geometry}
\usepackage{color,xcolor,soul}
\usepackage[abspage,user,savepos]{zref}
\usepackage[colorlinks, citecolor=blue]{hyperref}  
\usepackage{natbib}
\usepackage{fp}
\usepackage{framed}
\usepackage{cuted,balance}
\usepackage{setspace}
\usepackage{subcaption}
\usepackage{dingbat}
\usepackage{verbatim}
\usepackage{enumerate}
\usepackage{graphicx} 
\usepackage{epsfig} 
\usepackage{mathtools}
\usepackage{dsfont,amsmath} 
\usepackage{amssymb}  
\usepackage{authblk}
\usepackage{algorithm,algorithmic}
\usepackage{xcolor}
\usepackage{wasysym}
\usepackage{tikz}
\usetikzlibrary{shapes}
\usepackage[capitalize]{cleveref}
\usepackage{overpic}

\DeclareMathOperator*{\argmax}{argmax}
\newcounter{definition}
\newenvironment{definition}[1]{\refstepcounter{definition}\par\medskip
\noindent 
\textbf{Definition \thedefinition~(#1)} \em \rmfamily}
{\medskip}
\newenvironment{definition*}{\refstepcounter{definition}\par\medskip
\noindent 
\textbf{Definition \thedefinition.} \em \rmfamily}
{\medskip}

\newcounter{proposition}
\newenvironment{proposition}[1]{\refstepcounter{proposition}\par\medskip
\noindent 
\textbf{Proposition \theproposition~(#1)} \em \rmfamily}
{\medskip}
\newenvironment{proposition*}{\refstepcounter{proposition}\par\medskip
\noindent 
\textbf{Proposition \theproposition.} \em \rmfamily}
{\medskip}

\newenvironment{remark}
{\par\medskip\noindent\textbf{Remark }\em\rmfamily}
{\medskip}

\newenvironment{proof}{
   \indent \textit{Proof.} \rmfamily}{\hfill $\square$}

\newcommand{\BR}{\mathrm{BR}_{\tau}}
\newcommand{\uA}{D}
\newcommand{\ua}{d}
\newcommand{\oA}{U}
\newcommand{\oa}{u}

\begin{document}

\date{}

\title{\LARGE \bf
Omniscient Attacker in Stochastic Security Games with Interdependent Nodes}

\author{Yuksel Arslantas, Ahmed Said Donmez, Ege Yuceel and Muhammed O. Sayin
\thanks{Y. Arslantas, A. S. Donmez and M. O. Sayin are with the Department of Electrical \& Electronics Engineering at Bilkent University, Ankara, T\"{u}rkiye 06800. E. Yuceel is with he Department of Electrical and Computer Engineering, University of Illinois at
Urbana-Champaign, Champaign, IL, 61820.  (Emails: {\tt\small yuksel.arslantas@bilkent.edu.tr}, { \tt\small said.donmez@bilkent.edu.tr}, { \tt\small eyceel2@illinois.edu},  { \tt\small sayin@ee.bilkent.edu.tr})}%
}

\maketitle

\bigskip

\begin{center}
\textbf{Abstract}
\end{center}
The adoption of reinforcement learning  for critical infrastructure defense introduces a vulnerability where sophisticated attackers can strategically exploit the defense algorithm's learning dynamics. While prior work addresses this vulnerability in the context of repeated normal-form games, its extension to the stochastic games remains an open research gap. We close this gap by examining stochastic security games between an RL defender and an omniscient attacker, utilizing a tractable linear influence network model. To overcome the structural limitations of prior methods, we propose and apply neuro-dynamic programming. 
Our experimental results demonstrate that the omniscient attacker can significantly outperform a naive defender, highlighting the critical vulnerability introduced by the learning dynamics and the effectiveness of the proposed strategy. 

\begin{spacing}{1.245}

\section{Introduction} 
The growing sophistication of cyber-attacks against critical infrastructure, such as power grids or financial networks \citep{ref:Lehto22}, necessitates the development of innovative defense strategies. Traditional, static defense mechanisms struggle to adapt to attackers' ever-evolving tactics. Reinforcement learning (RL) offers a compelling solution to complex and dynamic threat landscapes by enabling agents to learn and adapt \citep{ref:Li19,ref:Adawadkar22}. However, alongside these benefits, a crucial aspect to consider is the vulnerability of RL algorithms themselves \citep{ref:Deng19,ref:Vundurthy23,ref:Arslantas24,ref:Bajaj24}. For example, an advanced adversary could exploit an RL-based security system by conducting reconnaissance \citep{ref:Mazurczyk21} or employing opponent modeling \citep{ref:Yu22} to learn the specific algorithm and its decision-making process. This knowledge could allow the attacker to bypass security measures or manipulate the system behavior for malicious purposes.

To develop robust RL-based defense structures, we need to understand how and to what extent these algorithms might be vulnerable to exploitation, where the RL algorithm becomes a hindrance to the agent due to differences in information structures and capabilities among the agents. For example in a competitive environment an agent with access to privileged information or greater computational power could exploit the other agent by compromising security, reducing the performance, or gaining an unfair advantage. Security games provide a valuable framework for analyzing this vulnerability and exploitation framework between a defender (security system) and an attacker (malicious actor) in various class of interactions \citep{ref:Sinha18} including cybersecurity setting \citep{ref:Alpcan10,ref:Etesami19}. Within such a framework, both players make strategic decisions to achieve their goals that are maximizing the security for the defender and bypassing defenses for the attacker. 

A significant body of work has leveraged this game-theoretic framework to investigate exactly how a strategically sophisticated opponent can exploit the predictable nature of common learning algorithms such as fictitious play \citep{ref:Vundurthy23}, Q-learning \citep{ref:Arslantas24}, experience-weighted attraction \citep{ref:Arslantas24b}, and no-regret learning \citep{ref:Deng19,ref:Guo25}. These studies commonly assumed an advanced agent with full knowledge of the other agents' learning dynamics and the underlying game structure. While the game structure varied, including zero-sum vs. general-sum games or simultaneous vs. sequential play, all of these works focused on the repeated play of normal-form games. As pointed out in the recent tutorial paper \cite{ref:Vamvoudakis25}, investigating the vulnerability of such algorithms in stochastic games remains a gap in the literature.

\begin{figure*}[t!]
    \centering
    \includegraphics[width=\columnwidth]{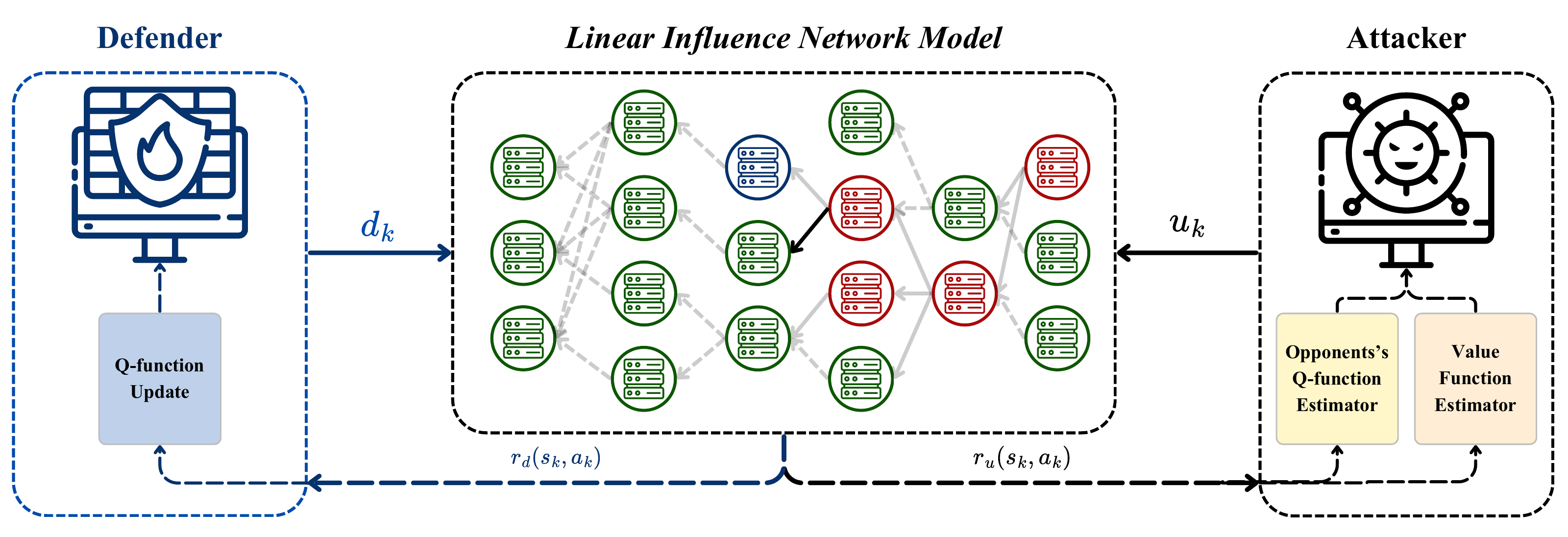}
    \caption{A linear influence network illustrating compromised and secured security assets. Green servers represent uncompromised assets, red servers represent compromised assets, and the blue server represents a defended asset. Shaded arrows depict interdependencies between assets: solid arrows indicate usable attack paths, and dashed arrows indicate unusable paths. The solid black arrow represents the attacker's direct attack.}\label{fig:model}
\end{figure*}

To address this gap, we examine the stochastic security games between a defender employing RL algorithm and a strategically sophisticated (omniscient) attacker. As a model of stochastic security game, we choose linear influence network model \citep{ref:Miura08,ref:Nguyen09,ref:Alpcan10}. This model captures the interdependencies among the security assets while providing a tractable state space, even for tabular learning dynamics. We extend the methodology presented in \cite{ref:Arslantas24}. In their work, \cite{ref:Arslantas24} construct a Markov decision process using the opponents' learning parameters within the context of repeated games and employ a quantization-based approximation scheme. While their method provides theoretical error bounds, extending it to settings with an increasing number of actions and multiple states necessitates an approximation structure different from quantization. Therefore, in this paper, we propose using one of the neuro-dynamic programming approaches which is approximate value iteration and demonstrate that the existing results for the repeated play of normal-form games can be generalized to stochastic games.

The paper is organized as follows: In Section \ref{sec:pre} we provide preliminary information on stochastic games and linear influence network models. In Section \ref{sec:problem} we formulate the problem and establish theoretical results. In Section \ref{sec:experiments} and \ref{sec:conclusion} we demonstrate the experimental results and conclude the paper, respectively.

\begin{figure*}[t!]
    \centering
    \includegraphics[width=\columnwidth]{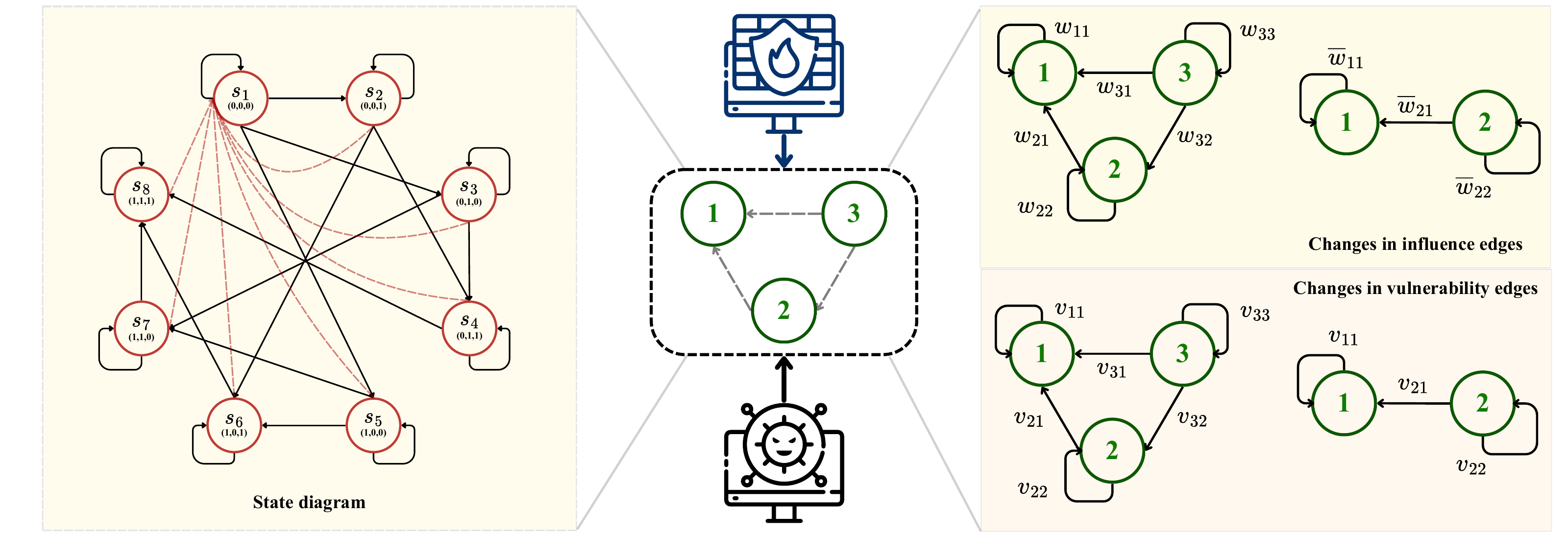}
    \caption{The illustration depicts the game model of a LIN. The left panel displays a state diagram involving three security assets, where solid lines denote state transitions based on defender and attacker actions, and dashed lines indicate a reset to the original state. The right panel illustrates the network's evolution, where edges represent the influence and vulnerability one asset exerts on another. When an asset is compromised, it is removed from the network and its influence edges are recomputed; however, securing the asset restores its original connections. Therefore, the influence matrix $I$ is stochastic. On the other hand, the vulnerability matrix $V$ is not stochastic because even a compromised asset continues to impact the vulnerabilities of others.}\label{fig:game}
\end{figure*}

\section{Preliminary Information}\label{sec:pre}
\subsection{Stochastic Games}
Consider a two-agent \textit{stochastic game} (SG) with \textit{finite} states and actions that can be characterized by the tuple $\langle S,(A^i,r^i,\gamma^i)_{i=1}^2,p\rangle$, where $S$ denotes the {finite} set of states, $A^i$ and $r^i:A\rightarrow\mathbb{R}$ for $A\coloneqq\prod_j A^j$ denote, resp., agent $i$'s action set and reward function. Furthermore, $p(\cdot\mid\cdot)$ is the transition function such that $p(s_+\mid s,a)$ for $(s,a,s_+)\in S\times A \times S$ is the probability of transition from $s$ to $s_+$ under action profile $a\in A$. 

At each stage $k=0,1,\ldots$, both agents observe the current state $s_k\in S$, take (possibly random) action $a_k^i\in A^i$ \textit{simultaneously}, then receives the reward $r_k^i = r^i(s_k,a_k)$ for $a_k = \{a_k^j\}_{j=1}^2$.
The state $s_k$ changes to $s_{k+1}$ with probability $p(s_{k+1}\mid s_k,a_k)$. The agent's goal to maximize the utility
\begin{flalign}\label{eq:utility}
    \mathrm{E}\left[\sum_{k=0}^{\infty}(\gamma^i)^k r^i(s_k,a_k)\right]
\end{flalign}
with the discount factor $\gamma^i \in (0,1)$, where the expectation is taken over the randomness on $(s_k,a_k)$.

If there is only one agent, say agent $i$, then SG reduces to \textit{Markov decision process} (MDP). Call the value of the state-action pair by the \textit{Q-function}. Based on the principle of optimality that the agent would play the best action in future stages, the Q-function, denoted by $q^i:S\times A^i\rightarrow\mathbb{R}$, is the (unique) solution to the following fixed-point equation \cite[Chapter 3]{ref:Sutton18}
\[
    q^i(s,a^i) = r^i(s,a^i) + \gamma\sum_{s_+\in S}p(s_+\mid s,a^i) \max_{\tilde{a}^i\in A^i} \{q^i(s_+,\tilde{a}^i)\}.
\]
Let $\Delta(A)$ denote the probability simplex over a finite set $A$. Then, any Markov stationary strategy $\pi^i:S\rightarrow\Delta(A^i)$ satisfying
$\pi^i(s) \in \argmax_{\mu^i}\mathrm{E}_{a^i\sim \mu^i}[q^i(s,a^i)]$ maximizes the utility \eqref{eq:utility}.

When the model of the underlying MDP (e.g., $r^i(\cdot)$ and $p(\cdot\mid\cdot)$) is unknown, the Q-learning algorithm, introduced by \cite{ref:Watkins92}, can recursively estimate the Q-function according to
\begin{subequations}\label{eq:Qupdate}
    \begin{flalign}
        &q_{k}^i(s_{k-1}^{},a_{k-1}^i) = (1-\alpha^i)q_{k-1}^i(s_{k-1}^{},a_{k-1}^i) + \alpha^i\big(r_{k-1}^i + \gamma^i \max_{\tilde{a}^i\in A^i}\{q_{k-1}^i(s_k,\tilde{a}^i)\}\big),\\
        &q_{k}^i(s,a^i) = q_{k-1}^i(s,a^i)\quad\forall (s,a^i)\neq(s_{k-1},a_{k-1}^i),
    \end{flalign}
\end{subequations}
where $\alpha^i\in (0,1)$ is some step size and $(s_{k-1},a_{k-1}^i)$ is the pair of the state and action at stage $k-1$. The action $a_k^i$ is taken according to the soft-max function $\BR:\mathbb{R}^{|A^i|}\rightarrow\Delta(A^i)$, defined by
\begin{flalign}\label{eq:soft}
    \BR(q_k^i(s_k,\cdot))(a^i) := \frac{\exp(q_k^i(s_k,a^i)/\tau)}{\sum_{\tilde{a}^i\in A^i}\exp(q_k^i(s_k,\tilde{a}^i)/\tau)}>0\quad\forall a^i
\end{flalign}
for some temperature parameter $\tau>0$ controlling the level of exploration, i.e., $a_k^i\sim \BR(q_k^i(s_k,\cdot))$. 

In SG, the other agent's play might be unknown similar to the model of the underlying environment. For such cases, agents can still follow the Q-update \eqref{eq:Qupdate} independently \textit{as if} there are no other agents or the other agents' play is a part of the underlying environment's (unknown) dynamics. This approach is known as \textit{independent Q-learning} (IQL), deployed widely in practice. 

\subsection{Linear Influence Network Models}
Linear influence network (LIN) models represent the relationships and interdependencies between agents and assets using a graph structure \citep{ref:Miura08,ref:Miura08b,ref:Nguyen09,ref:Zhou16}. These models allow for the analysis of how actions taken on one asset can affect the others. This underlying structure makes them well-suited for modeling security games, allowing security assets and their interdependencies to be represented as a graph, as illustrated in Fig. \ref{fig:model}.

Consider a weighted directed graph $\mathcal{G} = \{\mathcal{N},\mathcal{E}\}$, where $\mathcal{N}$ is the set of nodes, which we will call security assets, targeted by an attacker and defended by a defender, and $\mathcal{E}$ is the set of edges representing the connections. The weights on the edges represent the strength of the influence or vulnerability between assets. We can categorize the edges in the LINs into two types:
\begin{itemize}
    \item \textbf{Influence Edges ($\mathcal{E}_s$):} These edges indicate how much the security of one asset can be affected by the state of another. Specifically, we can define the \textit{influence matrix} as
    \begin{flalign}
        I = \begin{cases}
            w_{ij}, \quad &\text{if} \quad e_{ij} \in \mathcal{E}_s \\
            0, \quad &\text{otherwise},
        \end{cases}
    \end{flalign}
    where $0 < w_{ij} \leq 1$ $\forall i,j \in \mathcal{N}$ and $\sum_{i=1}^{|\mathcal{N}|} w_{ij} = 1$, $\forall j \in \mathcal{N}$.
    \item \textbf{Vulnerability Edges ($\mathcal{E}_v$):} These edges represent how exploiting one asset might make others vulnerable. We can define the \textit{vulnerability matrix} as follows
    \begin{flalign}
        V = \begin{cases}
            v_{ij}, \quad &\text{if} \quad e_{ij} \in \mathcal{E}_v \\
            0, \quad &\text{otherwise},
        \end{cases}
    \end{flalign}
    where $0 \leq v_{ij} \leq 1$ measures the vulnerability of asset $j$ caused by asset $i$ due to the interdependencies within the system. The support to node $j$ is defined to be $v_j \coloneqq \sum_{i}v_{ij}$.  
\end{itemize} 
Let \( x = \{x_1, x_2, \ldots, x_{|N|}\} \) be the vector measuring the value of independent security assets. These can be thought of as individual security assets without any connections as singletons. However, in a network structure, independent security assests lack significance. Therefore, we introduce the vector of effective security measures $ y = \{y_1, y_2, \ldots, y_{|N|}\} $ by using the influence matrix $ I $ such that
\[
y = Ix.
\]
This concept of influence plays a critical role in analyzing attacker and defender strategies within a LIN.

In the interaction between the attacker and the defender, the attacker can attack node \( i \), and the defender can defend node \( j \). Let \( u \) represent the attack by the attacker and \( d \) represent the defense by the defender. Furthermore, let the state of the network, \( s \), indicate which assets are compromised or not. 
We can define the probability that security asset $i$ is compromised as 
$$
p^i(s,u,d) = \begin{cases} 
    \max \left\{ p^i_{n0} (1-v_i) + p^i_{n1} v_i ,1 \right\} 
        & \hspace{0.25em} \text{if } u = i,\\
        & \hspace{0.25em} d \neq i,\\
        & \hspace{0.25em} s(i) = 0, \\
    \max \left\{ p^i_{d0} (1-v_i) + p^i_{d1} v_i ,1 \right\} 
        & \hspace{0.25em} \text{if } u = d = i,\\
        & \hspace{0.25em} s(i) = 0, \\
    0 
        & \hspace{0.25em} \text{otherwise}.
\end{cases}
$$
where $0 \leq p^i_{n1}\leq p^i_{n0} \leq 1$ be the standalone probability of compromise of node $i$, when there is full support and there is no support in the case defender does not defend node $i$ respectively. Similarly, $0 \leq p^i_{d1}\leq p^i_{d0} \leq 1$ are standalone compromise probabilities when there is full support and no support for node $i$, respectively, in the case that node $i$ is defended. Also, $s(i) = 0$ if the node $i$ is uncompromised and $s(i)=1$ if the node $i$ is compromised. 

The attacker's reward function $r(s,u,d)$ considers the effectiveness of compromising node $i$ within the current network state. Therefore, the reward function becomes
\begin{flalign}
    r(s,u,d) = p^i (s,u,d)y(s).
\end{flalign}
We can now formulate this attacker-defender interaction as a stochastic game, where the state of the network evolves based on the actions of both the attacker and defender, as well as the underlying probabilities \citep{ref:Nguyen09, ref:Alpcan10}. 

\section{Problem Formulation}\label{sec:problem}
Consider a two-agent SG $\mathcal{M}=\langle S,\uA, \oA, r_d, r_u, \gamma, p\rangle$ played by \textit{attacker} and \textit{defender} agents:
\begin{itemize}
    \item \textbf{Defender} agent is a \textit{naive} agent following the IQL algorithm \eqref{eq:Qupdate} based on its local information (i.e., state and reward received) as if there are no other agents. 
    \item \textbf{Attacker} agent is an \textit{advanced} decision-maker with the complete knowledge of the underlying game dynamics and the employed learning rule, and access to the joint actions of every player.
\end{itemize}
The state space $S$ indicates which security assets are compromised or not. We denote the action sets of the defender and attacker as $\uA$ and $\oA$, respectively. In LIN model, the defender can defend an uncompromised security asset and the attacker can attack, similarly, to an uncompromised one. 
The reward functions for each agent follow a similar notation to action sets. Specifically, $r_d(s,a)$ and $r_u(s,a)$ represent the reward function for the defender and the attacker, respectively, where $a \in A := D \times U$ is the joint action taken by the agents. $p(s_+|s,a)$ represents the probability of transition to state $s_+$ from $s$ when the joint action $a$ is taken.

At each stage $k$, defender takes action $\ua_k \sim \BR (q_k(s_k,\cdot))$ independently, depending on its Q-function estimate $q_k$. Attacker, knowing the employed learning dynamic, can track $q_k$ and compute the play of defender $\BR(q_k(s_k,\cdot))$. 
Therefore, the attacker can reformulate the game $\mathcal{M}$ as an MDP in which the defender is considered a non-strategic player and is modeled as part of the system.
\begin{definition}{Markov Decision Process}
    Given the SG $\mathcal{M}$, its MDP for the attacker can be characterized by the tuple $\mathcal{Z} = \langle Z, \oA, r_u, \overline{p},\gamma \rangle$. The state space $Z \coloneqq S \times Q$ is the compact set of states such that $Q \subset \mathbb{R}^{S \times \uA}$ corresponds to all possible Q-function estimates of the defender. The reward function $r_u: Z \times \oA \rightarrow \mathbb{R}$ is given by
    \begin{flalign}\label{eq:reward}
        r(z,\oa) = \mathrm{E}_{\ua\sim \BR(q)}[r_u(s,\oa,\ua)] \quad \forall (z,\oa)\in Z \times \oA.
    \end{flalign}
    The transition kernel $\overline{p}(\cdot|\cdot)$ determines the evolution of states according to transition probabilities $p(\cdot|\cdot)$ and the update rule \eqref{eq:Qupdate}. The discount factor $\gamma \in (0,1)$ is as described in $\mathcal{M}$.
\end{definition}

Given $(z,a)$, Q-function estimate of the defender $\widehat{q}$ can get updated to $\widehat{q}_+$ according to
\begin{subequations}
    \begin{flalign}
        &\widehat{q}_+(s,a) = (1-\alpha)\widehat{q}(s,a)+\alpha\left( r_d(s,a) + \gamma \max_{\tilde{a} \in \uA(s_+,\tilde{a})} \widehat{q}(s_+,a) \right),\\
        &\widehat{q}_+(\widehat{s},\widehat{a}) = \widehat{q}(\widehat{s},\widehat{a}) \quad \forall(\widehat{s},\widehat{a}) \neq (s,a),
    \end{flalign}
\end{subequations}
with probability $p(s_+|s,a)>0$.

With the complete knowledge of the employed IQL, attacker can compute the state $z_k \in Z$. Let the attacker follow stationary strategy $\pi: Z \rightarrow \Delta(\oA)$. Given the stationary strategy the attacker's utility \eqref{eq:utility} for the underlying SG $\mathcal{M}$ can be rewritten as 
\begin{flalign}
    \overline{U}(\pi) \coloneqq \mathrm{E}\left[ \sum_{k=0}^{\infty} \gamma^k r(z_k,\oa_k) \right]
\end{flalign}
for the MDP $\mathcal{Z}$ where the expectation is now taken over the randomness on $(z_k,\oa_k)$ due to \eqref{eq:reward}.
\begin{proposition}{Optimal Policy}
    There always exists an optimal stationary policy for $\mathcal{Z}$.
\end{proposition}

\begin{proof}
    The proof follows from \cite[Theorem 6.2.12]{ref:Puterman14} based on the observation that the (continuum) state space $Z = S \times Q$ in $\mathcal{Z}$ is a Polish space as $S$ is finite and $Q$ is a compact subset of $\mathbb{R}^{S \times \uA}$, and the action set $\uA$ is finite. 
\end{proof}

We can use value iteration to find the optimal policy, however, the continuous state space poses a challenge. Therefore, we employ neuro-dynamic programming (NDP), specifically approximate value iteration \citep{ref:Bertsekas96} that is defined by the following iteration
\begin{flalign}
    v_{\kappa+1} = \mathcal{A}\mathcal{T}v_\kappa,
\end{flalign}
where $\mathcal{A}$ is the approximation operator which can take various forms depending on the application such as quantization or parametric approximators and $\mathcal{T}$ is the Bellman operator. Specifically,
\begin{subequations}\label{eq:value}
    \begin{flalign}
        &(\mathcal{T}v_\kappa)(z) = \max_{u\in \oA} \left\{ r(z,u) + \gamma \int_Z v_\kappa(\tilde{z}) \overline{p}(d\tilde{z}|z,u) \right\},\\
        &(\mathcal{AT}v_\kappa)(z) = \max_{u\in \oA} \left\{ r(\widehat{z},u) + \gamma \int_{\widehat{Z}} v_\kappa(\tilde{z}) \overline{p}(d\tilde{z}|\widehat{z},u) \right\},
    \end{flalign} 
\end{subequations}
where $\widehat{Z} \subseteqq Z$ is the sampled state space for approximation and can be obtained via uniform or adaptive sampling.

\begin{proposition}{Approximation Error}\label{prop:error}
    Given the sampled state space $\widehat{Z}$, assume that there exists $\Delta \in \mathbb{R}_+$ such that $\min_{\widehat{z} \in \widehat{Z}}\| z - \widehat{z} \|_2 < \Delta$ for all $z \in Z$. Then, we have
    \begin{flalign}\label{eq:error}
        \|(\mathcal{T}v_\kappa)(z) - (\mathcal{AT}v_\kappa)(z)\|_\infty \leq \varepsilon,
    \end{flalign}
    where 
    \begin{flalign}
        \varepsilon := \frac{\Delta \sqrt{|\uA|}}{\tau (1-\gamma)^3} \max_{s,u,d} |r_u(s,u,d)|,
    \end{flalign}
    and $(\mathcal{T}v_\kappa)(z)$, $(\mathcal{AT}v_\kappa)(z)$ are as described in (\ref{eq:value}).
\end{proposition}

\begin{proof}
    The proof follows from \cite[Proposition 3]{ref:Arslantas24}.
\end{proof}

Having established a uniform bound on the single-step approximation error, we now address the stability of the value iteration process. The following proposition demonstrates that these errors do not diverge but result in a value function within a bounded neighborhood of the optimal value function.

\begin{proposition}{Convergence}\label{prop:convergence}
    Let $\varepsilon_\kappa := (\mathcal{T}v_\kappa)(z) - (\mathcal{AT}v_\kappa)(z)$ denote the approximation error at $\kappa$. Then,
    \begin{flalign}
        \limsup_{\kappa \rightarrow \infty} \|v^*(z) - (\mathcal{AT}v_\kappa)(z) \|_\infty \leq \frac{2\gamma}{(1-\gamma)^2}\varepsilon,
    \end{flalign} 
    where
    \begin{flalign}
        v^*(z) = (\mathcal{T}v^*)(z)
    \end{flalign}
    is the optimal value function of $\mathcal{Z}$ satisfying Bellman optimality equation.
\end{proposition}

\begin{proof}
    We have $\|\varepsilon_n\|_\infty \leq \varepsilon$ by Proposition \ref{prop:error}. The result then follows directly from \citep[Proposition 6.8]{ref:Bertsekas96}. 
\end{proof}

Proposition \ref{prop:convergence} establishes that our approximate value iteration converges to a bounded region of the optimal value function. This theoretical guarantee justifies the use of the approximation scheme in our practical implementation. 

\begin{remark}
    Although our empirical evaluation focuses on stochastic security games, the theoretical results developed in this section hold for any general-sum stochastic game in which one agent’s learning dynamics can be modeled as part of the environment. The structure of the analysis does not rely on security-specific properties; it only requires finitely many actions, a compact representation of the learning agent’s state, and well-defined transition dynamics. 
\end{remark}

\section{Experiments}\label{sec:experiments}
In the experiments, we used an LIN, defined similarly to the one in \cite{ref:Nguyen09}. This LIN consists of three security assets. The matrices $I$, $V$, and the vector $x$ are defined as follows
\begin{align}
    I =
    \begin{bmatrix}
    0.9 & 0.2 & 0 \\
    0 & 0.7 & 0 \\
    0.1 & 0.1 & 1
    \end{bmatrix} \quad
    V = 
    \begin{bmatrix}
    0.7 & 0 & 0 \\
    0.2 & 0.5 & 0 \\
    0.1 & 0.3 & 0.9
    \end{bmatrix} \quad
    x = 
    \begin{bmatrix}
    10 \\
    10 \\
    20
    \end{bmatrix}.
\end{align}

Each three security asset can either be compromised or not compromised by the attacker, meaning there are in total $8$ states as depicted in Fig. \ref{fig:game}. At each stage, both the attacker and the defender can choose one node to attack or defend or do nothing, respectively, meaning that each of them has $4$ actions. The state transition probabilities are defined so that $p_{n0}=0.7$, $p_{n1}=0.4$, $p_{d0} = 0.5$, $p_{d1} =0.2$, and $p_r = 0.2$, and $p_e = 0.3$, which are game reset and game end probabilities in case the attacker do nothing or fails. There are two models used for value function estimation. The models are neural networks that defines a sequence of fully connected layers with ReLU activations, designed for value function approximation with an input size of 36 (from each Q-value entry of the 4 actions across 9 total states, including the ``do nothing" action and the terminal state). The first complex model \textbf{CM} has hidden layers of 64, 64, and 32 neurons, and a vector output of size 9, that maps the Q-function to a value for each state, while the simpler model \textbf{SM} has no hidden layers and only an input and output layer. We train the NDP algorithm for 50 horizons, where in each horizon the model is trained for two epochs with $5 \times 10^6$ samples for training. The average MSE test error per horizon after two epochs is less than $10^{-2}$ for the model \textbf{CM} during the first few horizons and even less for the following horizons. For \textbf{SM} model, the training and testing errors are around $5 \times 10^{-1}$ for the first horizon and decreases below $10^{-2}$ after a few horizons.

In our experiments, we used learning rate $\alpha = 0.05$, discount factor $\gamma = 0.95$, and temperature parameter $\tau \in \{0.5, 1, 2\}$. We ran $50$ trials for each experiment with $2000$ episodes, averaging the total discounted reward per episode over the trials.
\begin{figure}[t!]
    \centering
    \includegraphics[width=.9\columnwidth]{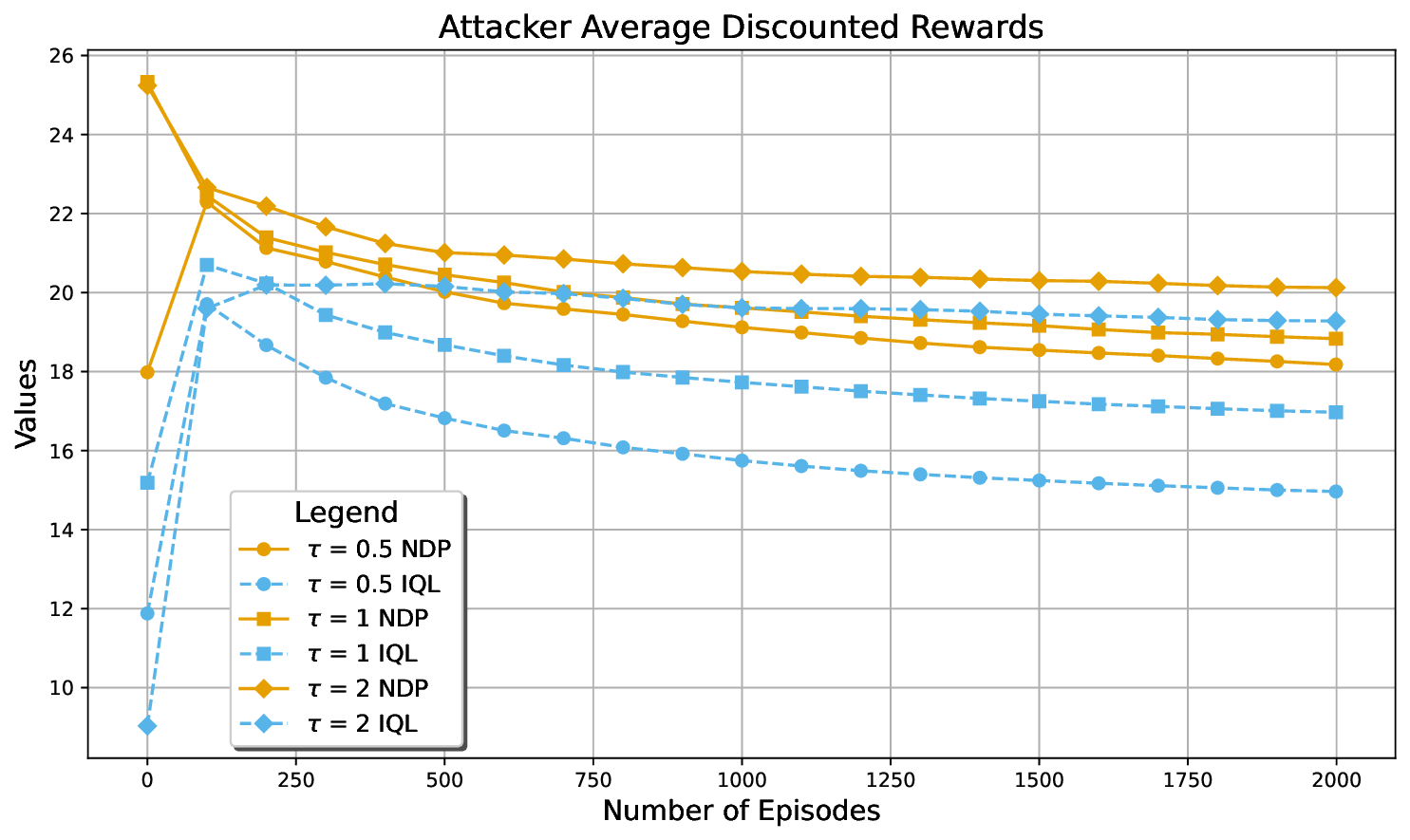}
    \caption{Average discounted rewards of the attacker against different exploration levels.}
    \label{fig:tau}
\end{figure}

\begin{figure}[t!]
    \centering
    \includegraphics[width=.9\columnwidth]{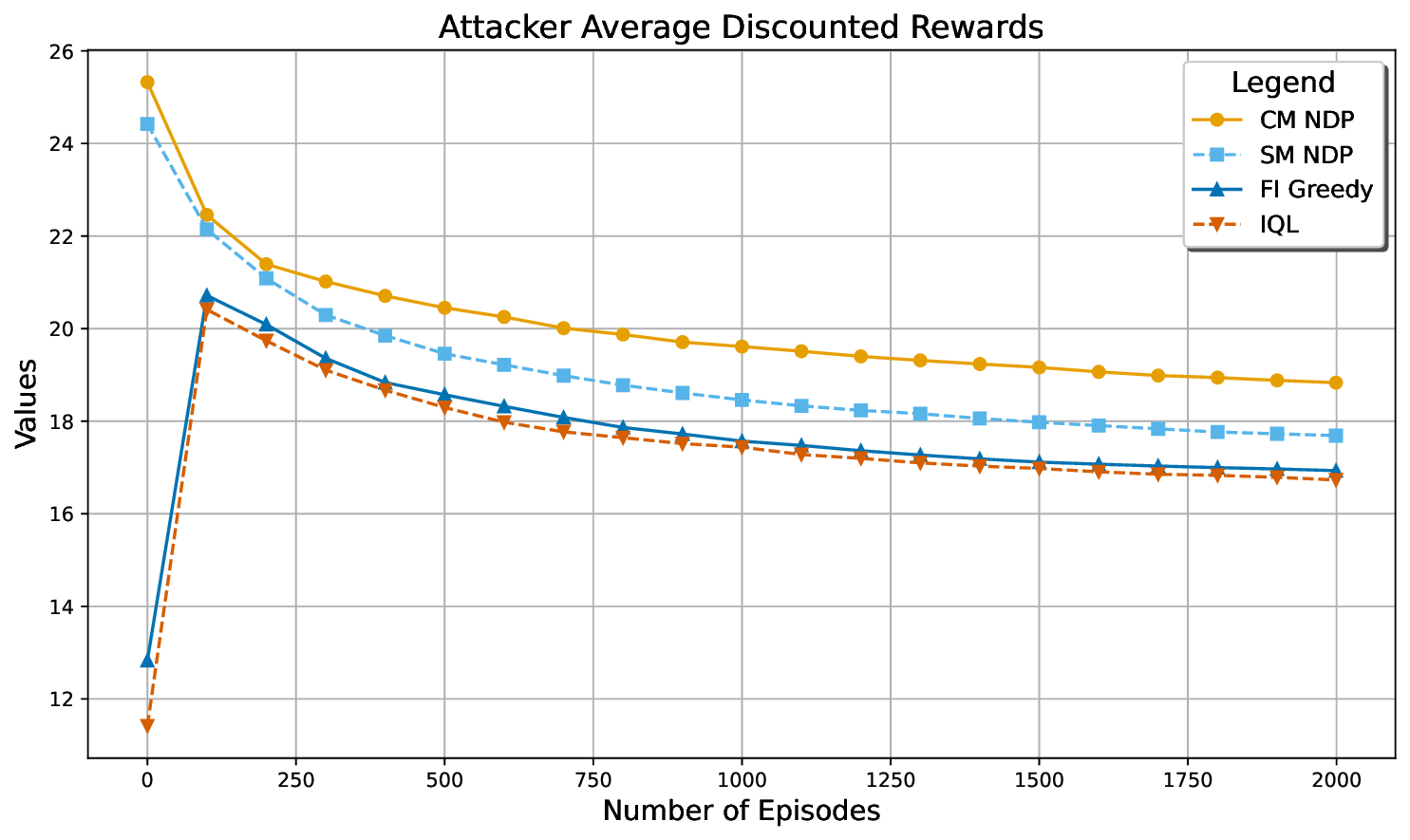}
    \caption{Average discounted rewards of the attacker with different learning models.}
    \label{fig:mod}
\end{figure}
We first investigated the impact of exploration on vulnerability. As shown in Fig. \ref{fig:tau}, the attacker consistently achieved a significantly higher total discounted reward compared to the defender. We observed that with increasing $\tau$ (more exploration), the reward for the attacker also increased. However, the gap between NDP attacker vs IQL defender and IQL attacker vs IQL defender decreased, indicating reduced vulnerability. 

Next, we evaluated the effect of the attacker's model complexity. Fig. \ref{fig:mod} shows that the \textbf{CM} outperformed the \textbf{SM} by achieving a higher average discounted reward. This demonstrates that the attacker's ability to approximate the defender's value function directly affects how much it can exploit the defender. Notably, both attacker models outperformed the greedy best response (\textbf{FI Greedy}) where the agent utilizes knowledge of the opponent's $q$ values and the game model to play in a way that maximizes the immediate reward given this information, further supporting the dominance of the NDP algorithm.

\section{Conclusion}\label{sec:conclusion}
In this work we investigated the challenge of analyzing how vulnerable RL agents are in stochastic security games. We specifically modeled the interaction between a defender employing independent Q-learning and an omniscient attacker in linear influence networks to capture the essential interdependencies among security assets.
Our primary contribution involved adapting the analysis of learning dynamic exploitation from repeated normal-form games to stochastic games through neuro-dynamic programming. This methodology allowed the attacker to approximate the defender's Q-value function, enabling  manipulation of the defender's learning process.
The experimental results provided compelling and consistent evidence regarding the extent of this vulnerability. Across various settings, the omniscient attacker consistently achieved a higher average discounted reward compared to the defender, confirming that the attacker can strategically exploit the defender's learning dynamics within the stochastic game environment.
This vulnerability highlights the necessity for future research to focus on developing robust and resilient RL-based defense mechanisms that are less susceptible to such sophisticated, model-aware manipulation.

\section{Acknowledgements}
This work was supported by The Scientific and Technological Research Council of Türkiye (TUBITAK) BIDEB 2211-A National PhD scholarship program.

\end{spacing}
\newpage
\begin{spacing}{1}
\bibliographystyle{plainnat}
\bibliography{mybib}
\end{spacing}

\end{document}